\documentclass{article}
\usepackage{fullpage}
\usepackage[american]{babel}
\usepackage{amssymb}
\usepackage{amsmath}
\usepackage{algorithmic}
\usepackage{multirow}
\usepackage{graphicx}
\usepackage{amsthm} 
\usepackage{enumitem}
\usepackage[ruled,vlined]{algorithm2e}

\newcommand{\eqn}[1]{(\ref{#1})}
\newcommand{\beql}[1]{\begin{equation}\label{#1}}
\newcommand{\eeq}{\end{equation}}
\newtheorem{theo}{Theorem}
\newtheorem{defi}{Definiton}

\newtheorem{lem}{Lemma}

\title{A simple linear space algorithm for computing a longest common increasing subsequence}
\author{Daxin Zhu, Lei Wang, Tinran Wang, and Xiaodong Wang}

\begin{document}
\maketitle

\begin{abstract}
This paper reformulates the problem of finding a longest common increasing subsequence of the two given input sequences in a very succinct way. An extremely simple linear space algorithm based on the new formula can find a longest common increasing subsequence of sizes $n$ and $m$ respectively, in time $O(nm)$ using additional $\min\{n,m\}+1$ space.
\end{abstract}

\section{Introduction}

The study of the longest common increasing subsequence (LCIS) problem originated from two classical subsequence problems, the longest common subsequence (LCS) and the longest increasing subsequence (LIS).
The classic algorithm to the LCS problem is the dynamic programming solution of Wagner and Fischer\cite{9}, with $O(n^2)$ worst case running time. Masek and Paterson\cite{7} improved this algorithm by using the "Four-Russians" technique to reduce its running time to $O(n^2/\log n)$ in the worst case. Since then, there has been not much improvement on the time complexity in terms of $n$ found in the literature.
There is also a rich history for the longest increasing subsequence problem, e.g., see \cite{2,3,5}.

The LCIS problem for input sequences $X$ and $Y$ consists of finding a subsequence $Z$ of $X$ and $Y$ which is both an increasing sequence and a common subsequence of $X$ and $Y$ with maximal length.
Yang et al. \cite{10} designed a dynamic programming algorithm that finds an LCIS of two input sequences of size $n$ and $m$ in $O(nm)$ time and space.
If the length of the LCIS, $l$, is small, Katriel and Kutz \cite{6} gave a faster algorithm which runs in $O(nl\log n)$ time.
If $r$, the total number of ordered pairs of positions at which the two input sequences match, is  relatively small, Chan et al. \cite{1} gave a faster algorithm which runs in $O(\min(r\log l,nl + r)\log\log n + n\log n)$ time where $n$ is the length of each sequence and $r$ is the total number of ordered pairs of positions at which the two sequences match and $l$ is the length of the LCIS.
A first linear space algorithm was proposed by Yoshifumi Sakai \cite{8}. The space cost of the algorithm of Yang et al. was reduced to linear by a careful application of Hirschberg¡¯s divide-and-conquer method \cite{4}. The space complexity of the algorithm of Katriel and Kutz \cite{6} was also reduced from $O(nl)$ to $O(m)$ by using the same divide-and-conquer method of Hirschberg \cite{4}.

In this paper, we solve the problem in a new insight. Based on a novel recursive formula, we find a very simple linear space algorithm but not the Hirschberg¡¯s divide-and-conquer method to solve the problem.

\section{Definitions and Terminologies}

In the whole paper we will use $X=x_1x_2\cdots x_n$ and $Y=y_1y_2\cdots y_m$ to denote the two input sequences of size $n$ and $m$ respectively, where each pair of elements in the sequences is comparable.

Some terminologies on the LCIS problem are usually referred to in the following way.

\begin{defi}\label{df1}\hfill

A subsequence of a sequence $X$ is obtained by deleting zero or more characters from $X$ (not necessarily contiguous).
For a given sequence $X=x_1x_2\cdots x_n$ of length $n$, the $i$th character of $X$ is denoted as $x_i \in \sum$ for any $i=1,\cdots,n$.
We set the $i$th prefix of $X=x_1x_2\cdots x_n$, for $i=0,1,\cdots,n$, as $X_i=x_1x_2\cdots x_i$ , and $X_0$ is the empty sequence.
\end{defi}

\begin{defi}\label{df2}\hfill

An appearance of sequence $Z=z_1z_2\cdots z_k$ in sequence $Y=y_1y_2\cdots y_m$, starting at position $j$ is a sequence of strictly increasing indexes $j_1,j_2,\cdots,j_k$ such that $j_1=j$, and $Z=y_{j_1},y_{j_2},\cdots,y_{j_k}$.
The sequence $Z$ is referred to as a subsequence of $Y$.

Given two sequences $X=x_1x_2\cdots x_n$ and $Y=y_1y_2\cdots y_m$, we say that a sequence $Z$ is a common subsequence of $X$ and $Y$ if $Z$ is a subsequence of both $X$ and $Y$.
For the given sequence $X=x_1x_2\cdots x_n$, if $x_1<x_2<\cdots <x_n$, then $X$ is called an increasing
sequence.
For the two sequences $X=x_1x_2\cdots x_n$ and $Y=y_1y_2\cdots y_m$, we say that a sequence $Z$ is a common increasing subsequence (CIS) of $X$ and $Y$ if $Z$ is an increasing sequence and a common subsequence of $X$ and $Y$.

The longest common increasing subsequence of $X$ and $Y$, is a common increasing subsequence whose length is the longest among all common increasing subsequences of the two given sequences.
\end{defi}

\noindent{\bf Example.}\hfill

Let $X=(3,5,1,2,7,5,7)$ and $Y=(3,5,2,1,5,7)$. We have, $n=7$ and $m=6$. $X_3=(3,5,1)$, and $X_0$ is the empty sequence.
$Z=(3,1,2,5)$ is a subsequence of $X$  with corresponding index sequence $(1,3,4,6)$.

The subsequence $(3,5,1)$ and $(3,5,7)$ are common subsequences of both $X$ and $Y$, and the subsequence $(3,5,7)$ is an LCIS of $X$ and $Y$.

\begin{defi}\label{df3}\hfill
For each pair $(i,j),1\leq i\leq n, 1\leq j\leq m$, the set of all LCISs of $X_i$ and $Y_j$ that ends on $y_j$ is denoted by $LCIS(X_i,Y_j)$. The length of an LCIS in $LCIS(X_i,Y_j)$ is denoted as $f(i,j)$.
\end{defi}

\begin{defi}\label{df5}\hfill

A match for two sequences $X=x_1x_2\cdots x_n$ and $Y=y_1y_2\cdots y_m$ is an ordered pair $(i,j),1\leq i\leq n, 1\leq j\leq m$ such that $x_i=y_j$.
The match function $\delta(i,j)$ of $X$ and $Y$ can be defined as:
\beql{eq21}
\delta(i,j)=
\begin{cases}
1, & \text{if $x_i=y_j$} \\
0, & \text{otherwise.}
\end{cases}
\eeq
\end{defi}

\begin{defi}\label{df6}\hfill

For each pair $(i,j), 1\leq i\leq n, 1\leq j\leq m$, the index set $\beta(i,j)$ can be defined as follows:

\beql{eq22}
\beta(i,j)=\{t\mid 1\leq t<j, y_t<x_i\}
\eeq
\end{defi}

\section{A recursive formula}

Similar to the $O(nm)$ solution of Wagner and Fischer for computing the length of an LCS, a standard dynamic programming
algorithm can be built based on the following recurrence for the length $f(i,j)$ of an LCIS in $LCIS(X_i,Y_j)$, $1\leq i\leq n, 1\leq j\leq m$.

\begin{theo}\label{th1}\hfill

Let $X=x_1x_2\cdots x_n$ and $Y=y_1y_2\cdots y_m$ be two input sequences over an alphabet $\sum$ of size $n$ and $m$ respectively.
For each $1\leq i\leq n, 1\leq j\leq m$, $f(i,j)$, the length of an LCIS of $X_i$ and $Y_j$ that ends on $y_j$, can be computed by the following dynamic programming formula.

\beql{eq31}
f(i,j)=
\begin{cases}
0, & \text{if $i=0$  or $j=0$}, \\
f(i-1,j), & \text{if $i,j>0$ and $x_i\neq y_j$,}\\
1+\displaystyle\max_{t\in\beta(i,j)} f(i-1,t), & \text{if $i,j>0$ and $x_i=y_j$}.
\end{cases}
\eeq

\end{theo}

\begin{proof}\hfill

(1) The initial case is trivial.

(2) In the case of $x_i\neq y_j$, we have,
$Z\in LCIS(X_i,Y_j)$ if and only if $Z\in LCIS(X_{i-1},Y_j)$, and thus $LCIS(X_i,Y_j)=LCIS(X_{i-1},Y_j)$. Therefore, $f(i,j)=f(i-1,j)$.

(3) In the case of $x_i=y_j$, let $Z=z_1z_2\cdots z_k\in LCIS(X_i,Y_j)$ be an an LCIS of $X_i$ and $Y_j$ that ends on $y_j$. In this case, we have, $f(i,j)=k$, and $z_1z_2\cdots z_{k-1}$ must be a common increasing subsequence of $X_{i-1}$ and $Y_t$ for some $1\leq t<j$, and $z_{k-1}=y_t<y_j$. It follows that $k-1\leq LCIS(X_{i-1},Y_t)$, and thus

\beql{eq32}
f(i,j)\leq 1+\displaystyle\max_{t\in\beta(i,j)} f(i-1,t)
\eeq

On the other hand, let $Z=z_1z_2\cdots z_k\in LCIS(X_{i-1},Y_t)$ for some $1\leq t<j$, and $z_{k}=y_t<y_j$, then $Z\oplus y_j$ must be a common increasing subsequence of $X_{i}$ and $Y_j$ ending on $y_j$. This means, $k+1\leq f(i,j)$, and thus $f(i-1,t)+1\leq f(i,j)$. It follows that
\beql{eq33}
f(i,j)\geq 1+\displaystyle\max_{t\in\beta(i,j)} f(i-1,t)
\eeq

Combining \eqn{eq32} and \eqn{eq33}, we have $f(i,j)=1+\displaystyle\max_{t\in\beta(i,j)} f(i-1,t)$.

The proof is complete.
\end{proof}

The table $f$ has a very nice property as stated in the following Lemma.

\begin{lem}\label{lm1}\hfill

For each pair $(i,j), 1\leq i\leq n, 1< j\leq m$, if $\delta(i,j)=1$ and $f(i,j)>1$, then there must be an index $r$ such that

\beql{eq37}
\begin{cases}
1\leq r<j,\\
y_r<y_j,\\
f(i,r)=f(i,j)-1.
\end{cases}
\eeq
\end{lem}

\begin{proof}\hfill

It follows from \eqn{eq31} and $\delta(i,j)=1$ that $f(i,j)=1+\displaystyle\max_{t\in\beta(i,j)} f(i-1,t)=1+f(i-1,r)$, where $1\leq r<j$ and $y_r<y_j$.
It follows from \eqn{eq31} and $\delta(i,r)=0$ that $f(i,r)=f(i-1,r)=f(i,j)-1$.

The proof is complete.
\end{proof}

\section{Implementations}
Based on Theorem \ref{th1}, the length of LCISs for the given input sequences $X=x_1x_2\cdots x_n$ and $Y=y_1y_2\cdots y_m$ of size $n$ and $m$ respectively, can be computed in $O(nm)$ time and $O(nm)$ space by a standard dynamic programming algorithm.

\begin{algorithm}[H]
\SetAlgoLined
\KwIn{$X,Y$}
\KwOut{$f(i,j)$, the length of LCIS of $X_i$ and $Y_j$ ending on $y_j$}
\For{i=1 \emph{\KwTo} n}{
$\theta \leftarrow 0$\;
\For{j=1 \emph{\KwTo} m}{
$f(i,j) \leftarrow f(i-1,j)$\;
\lIf{$x_i>y_j$ {\bf and} $f(i,j)>\theta$}
{$\theta \leftarrow f(i,j)$\;}
\lIf{$\delta(i,j)=1$}
{$f(i,j) \leftarrow \theta+1$\;}
}}
\KwRet {$\max_{1\leq j\leq m} f(n,j)$}
\caption{LCIS}
\end{algorithm}

It is clear that the time and space complexities of the algorithm are both $O(nm)$.

When computing a particular row of the dynamic programming table, no rows before the previous row are required. Thus only two rows have to be kept in memory at a time. Without loss of generality, we can assume $n\geq m$ in the following discussion. Thus, we need only $\min\{n,m\}+1$ entries to compute the table.
Based on Hirschberg¡¯s divide-and-conquer method of solving the LCS problem in linear space \cite{6}, Yoshifumi Sakai presented a linear space algorithm for computing an LCIS. The algorithm is a bit involved. Based on the formula (\ref{eq31}), we can reduce the space cost of the algorithm $LCIS$ to $\min\{n,m\}+1$. The improved linear space algorithm can also produce an LCIS in adiitional $O(m)$ time.

A space efficient algorithm to compute $f(i,j)$ and an LCIS can be described as follows.

\begin{algorithm}[H]
\SetAlgoLined
\caption{Linear Space}
\KwIn{$X,Y$}
\KwOut{$f(i,j)$, the length of LCIS of $X_i$ and $Y_j$ ending on $y_j$}
\For{$i=1$ to $n$}{
$L(0) \leftarrow 0$\;
\For{$j=1$ to $m$}{
\lIf{$x_i>y_j$ {\bf and} $L(0)<L(j)$}
$L(0) \leftarrow L(j)$\;
\lIf{$x_i=y_j$}
{$L(j) \leftarrow 1+L(0)$\;}
}
}
$L(0)\leftarrow\max_{1\leq j\leq m} L(j)$\;
\KwRet {$L(0)$}
\end{algorithm}

In the algorithm above, the array $L$ of size $m+1$ is utilized to hold the appropriate entries of $f$.
At the time $f(i,j)$ to be computed, $L$ will hold the following entries:
\begin{itemize}
\item $L(k)=f(i,k)$ for $1\leq k <j-1$ (i.e., earlier entries in the current row);

\item $L(k)=f(i-1,k)$ for $k\geq j-1$ (i.e., entries in the previous row);

\item $L(0)=\theta$ (i.e., the previous entry computed, which has a maximal value).
\end{itemize}

Therefore, a total of $m+1$ entries is used in the algorithm. The time complexity of the algorithm is obviously $O(nm)$.
At the end of the algorithm, the maximal length is stored in $L(0)$. It follows from Lemma 1, we can produce an LCIS of $X$ and $Y$ by using the computed array $L$ as follows.
The LCIS is produced backwards. The elements can be found successively by a recursive scan algorithm $Next$ as follows.

\begin{algorithm}[H]
\SetAlgoLined
\caption{$Next(j,v,len)$}
\KwIn{The current position $j$; the last element $v$; the current length $len$.}
\KwOut{The next element.}
\lWhile{$L(j)\neq len$ {\bf or} $v\leq y_j$}
{$j \leftarrow j-1$\;}
\If{$j>0$ {\bf and} $len>0$}
{$Next(j,y_j,len-1)$\;
{\bf Print} $y_j$\;
}
\end{algorithm}

The LCIS can then be produced by an initial call $Next(m,\infty,L(0))$.
It is clear that the algorithm will produce an LCIS of $X$ and $Y$ in additional $O(m)$ time.

\noindent{\textbf{Example} ( \textit{continued} ).}\hfill

For the given input sequences of $X=(3,5,1,2,7,5,7)$ and $Y=(3,5,2,1,5,7)$,
$f(i,j)$, the length of an LCIS in $LCIS(X_i,Y_j)$ is listed below.
\[
f=
\begin{pmatrix}
1, 0, 0, 0, 0, 0\\
1, 2, 0, 0, 2, 0\\
1, 2, 0, 1, 2, 0\\
1, 2, 1, 1, 2, 0\\
1, 2, 1, 1, 2, 3\\
1, 2, 1, 1, 2, 3\\
1, 2, 1, 1, 2, 3
\end{pmatrix}
\]

It follows from the above that the length of any LCIS of $X$ and $Y$ is 3.
The LCIS $(1,5,7)$ of $X$ and $Y$ can be generated by the algorithm $Build$.

Finally, our main result can be completed in the following theorem.

\begin{theo}\label{th6}\hfill

Let $X=x_1x_2\cdots x_n$ and $Y=y_1y_2\cdots y_m$ be two input sequences over an alphabet $\sum$ of size $n$ and $m$ respectively.
A longest common increasing subsequences of $X$ and $Y$ can be computed in time $O(nm)$ using additional $\min\{n,m\}+1$ space.

\end{theo}

\section{Concluding remarks}
We have reformulated the problem of computing a longest common increasing subsequence of the two given input sequences $X$ and $Y$ of size $n$ and $m$ respectively. An extremely simple linear space algorithm based on the new formula can find a longest common increasing subsequence of $X$ and $Y$ in time $O(nm)$ using additional $\min\{n,m\}+1$ space. The time complexity of the new algorithm may be improved further.

\end{document}